\documentclass[twoside]{aiml}

\usepackage{aimlmacro}
\usepackage{amsmath}
\usepackage{amssymb}
\usepackage{capt-of}
\usepackage{enumitem}
\usepackage{graphicx}
\usepackage{hyperref}
\usepackage{multicol}
\usepackage{multirow}
\usepackage{stmaryrd}
\usepackage{subcaption}
\usepackage{thm-restate}
\usepackage{tikz}
\usepackage{xspace}

\usetikzlibrary{arrows.meta,babel,bending,calc,decorations.pathreplacing,calligraphy,matrix,positioning}

\newcommand{\FOL}{\textrm{\upshape FO}\xspace}

\newcommand{\FOtwo}{\textrm{\upshape FO$^2$}\xspace}
\newcommand{\FOk}{\textrm{\upshape FO$^k$}\xspace}
\newcommand{\Ctwo}{\textrm{\upshape C$^2$}\xspace}
\newcommand{\Ck}{\textrm{\upshape C$^k$}\xspace}

\newcommand{\BML}{\textrm{\upshape ML}\xspace}
\newcommand{\PBML}{\textrm{\upshape ML$^+$}\xspace}
\newcommand{\PML}{\textrm{\upshape ML$^+_\Diamond$}\xspace}
\newcommand{\PMLk}{\textrm{\upshape ML$^{+,k}_\Diamond$}\xspace}
\newcommand{\PMLb}{\textrm{\upshape ML$^{+,B}_\Diamond$}\xspace}
\newcommand{\PMLbk}{\textrm{\upshape ML$^{+,B,k}_\Diamond$}\xspace}
\newcommand{\PMLg}{\textrm{\upshape ML$^{+,G}_\Diamond$}\xspace}

\newcommand{\GML}{\textrm{\upshape ML$_{\#}$}\xspace}
\newcommand{\GMLk}{\textrm{\upshape ML$^k_{\#}$}\xspace}
\newcommand{\GMLb}{\textrm{\upshape ML$^B_{\#}$}\xspace}
\newcommand{\GMLbk}{\textrm{\upshape ML$^{B,k}_{\#}$}\xspace}
\newcommand{\GMLg}{\textrm{\upshape ML$^G_{\#}$}\xspace}
\newcommand{\HL}{\textrm{\upshape HL$(\downarrow,@)$}\xspace}
\newcommand{\HLe}{\textrm{\upshape HL$(\textrm{\upshape E}\xspace,\downarrow,@)$}\xspace}
\newcommand{\HLb}{\textrm{\upshape HL$^B(\downarrow,@)$}\xspace}

\newcommand{\homN}{\textrm{hom}_{\mathbb{N}}}
\newcommand{\homB}{\textrm{hom}_{\mathbb{B}}}
\newcommand{\homS}{\textrm{hom}_{\mathcal{S}}}

\newcommand{\Hom}{\textrm{Hom}}
\newcommand{\Ext}{\textrm{Ext}}
\newcommand{\Inj}{\textrm{Inj}}
\newcommand{\Sur}{\textrm{Sur}}
\newcommand{\cnt}{\textrm{count}}
\newcommand{\cntS}{\textrm{count}_\mathcal{S}}

\newcommand{\Prop}{\textrm{Prop}}
\newcommand{\Act}{\mathbb{A}}
\newcommand{\Succ}{\textrm{Succ}}
\newcommand{\Pred}{\textrm{Pred}}
\newcommand{\markform}{\textrm{mark}}
\newcommand{\marking}{\lambda}

\newcommand{\dback}{\blacklozenge}
\newcommand{\glob}{\textrm{\upshape E$^{\geq k}$}\xspace}
\newcommand{\tr}{\textrm{tr}}
\newcommand{\wvar}{\textrm{WVAR}}

\newcommand{\parent}{\textrm{parent}}
\newcommand{\depth}{\textrm{depth}}

\newcommand{\logic}{\mathcal{L}}
\newcommand{\TauModels}{\mathcal{M}^n_\tau}

\newcommand{\KT}{\mathcal{T}}
\newcommand{\GK}{\mathcal{PG}}
\newcommand{\AKT}{\mathcal{A}}
\newcommand{\FKT}{\mathcal{F}}

\newcommand{\MLequiv}{\equiv_{ML}}
\newcommand{\PMLequiv}{\equiv_{\PML}}
\newcommand{\PMLkequiv}{\equiv_{\PMLk}}

\newcommand{\PMLbkequiv}{\equiv_{\PMLbk}}
\newcommand{\PMLgequiv}{\equiv_{\PMLg}}

\newcommand{\PBMLequiv}{\equiv_{\PBML}}
\newcommand{\GMLequiv}{\equiv_{\GML}}
\newcommand{\GMLkequiv}{\equiv_{\GMLk}}

\newcommand{\GMLbkequiv}{\equiv_{\GMLbk}}
\newcommand{\GMLgequiv}{\equiv_{\GMLg}}
\newcommand{\HLequiv}{\equiv_{\HL}}
\newcommand{\HLbequiv}{\equiv_{\HLb}}
\newcommand{\bisim}{\leftrightarroweq}

\newcommand{\dirsim}{\leftrightarroweq_d}

\newcommand{\unr}{\texttt{unr}}
\newcommand{\gsub}{\texttt{gsub}}
\newcommand{\expfunc}{\texttt{exp}}
\newcommand{\expand}[1]{\texttt{exp}(#1)}
\newcommand{\expandk}[1]{\texttt{exp}^k(#1)}
\newcommand{\flip}{\texttt{flip}}
\newcommand{\inst}{\texttt{inst}}

\newcommand{\last}{\texttt{last}}
\newcommand{\reach}[1]{\mathcal{R}^{#1}}
\newcommand{\unreach}[1]{\mathcal{U}^{#1}}
\newcommand{\rec}{\mathcal{P}}
\newcommand{\PG}{\mathcal{PG}}

\newcommand{\el}{\textrm{el}}
\newcommand{\dom}{\textrm{dom}}
\newcommand{\rng}{\textrm{rng}}
\newcommand{\tup}[1]{\overline{#1}}
\newcommand{\C}{\mathcal{C}}

\def\lastname{Comer}

\begin{document}

\begin{frontmatter}
\title{Lov\'{a}sz Theorems for Modal Languages}
\author{Jesse Comer}
\address{University of Pennsylvania}

\begin{abstract}
A famous result due to Lov\'{a}sz states that two finite relational structures $M$ and $N$ are isomorphic if, and only if, for all finite relational structures $T$, the number of homomorphisms from $T$ to $M$ is equal to the number of homomorphisms from $T$ to $N$. Since first-order logic ($\FOL$) can describe finite structures up to isomorphism, this can be interpreted as a characterization of $\FOL$-equivalence via \emph{homomorphism-count indistinguishability} with respect to the class of finite structures. We identify classes of labeled transition systems (LTSs) such that homomorphism-count indistinguishability with respect to these classes, where ``counting'' is done within an appropriate semiring structure, captures equivalence with respect to positive-existential modal logic, graded modal logic, and hybrid logic, as well as the extensions of these logics with either backward or global modalities. Our positive results apply not only to finite structures, but also to certain well-behaved infinite structures. We also show that equivalence with respect to positive modal logic and equivalence with respect to the basic modal language are not captured by homomorphism-count indistinguishability with respect to any class of LTSs, regardless of which semiring is used for counting.
\end{abstract}

\begin{keyword}
Homomorphism, Semiring, Graded Modal Logic, Hybrid Logic.
\end{keyword}
\end{frontmatter}

\section{Introduction}
\label{sec:intro}
Lovász's theorem \cite{lovasz1967operations} grew out of the study of a fundamental computational problem in graph theory and complexity theory: the graph isomorphism problem. This problem is significant because it is not known to be solvable in polynomial time, but is also not known to be \textrm{\upshape \bf{NP}}-complete. In fact, recent work has shown that the problem can be resolved in quasipolynomial time \cite{babai2016graph}, and it is considered to be a potential member of the conjectured class of \textrm{\upshape \bf{NP}}-intermediate problems, which exist if and only if $\textrm{\upshape \bf{P}} \neq \textrm{\upshape \bf{NP}}$ \cite{ladner1975structure}. Due to the high running time of known exact algorithms for the problem, and the difficulty in determining a lower bound on its complexity, researchers have turned toward the study of
heuristic algorithms, such as the color-refinement algorithm, which can distinguish many (but not all) non-isomorphic graphs \cite{babai1980random}.

The Lov\'{a}sz theorem relates homomorphisms to isomorphisms; while originally stated for structures with a single relation of arbitrary finite arity, it will be convenient for our purposes to consider its generalization to arbitrary finite relational structures. A map between two finite relational structures is a \emph{homomorphism} if, whenever a tuple of elements in the first structure occurs in some relation, then the image of that tuple must also occur in the corresponding relation in the second structure. Given finite structures $M$ and $N$, we write $\homN(N,M)$ to denote the number of homomorphisms from $N$ to $M$. Given a class $\mathcal{C}$ of finite structures and a fixed finite structure $M$, we can form the \emph{homomorphism count vector} of $M$ with respect to the class $\mathcal{C}$: the sequence $\homN(\mathcal{C},M) = \langle \homN(A,M) \rangle_{A \in \mathcal{C}}$. Using this notation, Lov\'{a}sz's result in \cite{lovasz1967operations} can be stated as follows: two finite relational structures $M$ and $N$ are isomorphic if and only if $\homN(\mathcal{M},M) = \homN(\mathcal{M},N)$, where $\mathcal{M}$ is the class of all finite structures. Informally, this says that homomorphism count indistinguishability with respect to $\mathcal{M}$ \emph{captures} isomorphism between finite structures.

Every class of finite structures $\mathcal{C}$ induces an equivalence relation $\sim_\mathcal{C}$ on finite structures defined by $M \sim_\mathcal{C} N$ if and only if $\homN(\mathcal{C},M) = \homN(\mathcal{C},N)$. Dvo\v{r}ák initiated the study of such equivalence relations for proper subclasses $\mathcal{C}$ of $\mathcal{M}$, showing that two undirected graphs are homomorphism count indistinguishable with respect to the class of trees if and only if they are indistinguishable by the color-refinement algorithm \cite{dvovrak2010recognizing}. This was later proven independently by Dell et. al. \cite{dell2018lov}. In fact, Dvo\v{r}ák and Dell et. al. proved a more general result: homomorphism count indistinguishability with respect to graphs of tree-width at most $k$ captures indistinguishability by the $k$-dimensional Weisfeiler-Leman (WL) method, where the color-refinement algorithm is the special case for $k=1$.

Given two graphs with adjacency matrices $A$ and $B$, an isomorphism between them can be interpreted as a permutation matrix $X$ such that $AX = B$. If we drop the requirement that $X$ contain only binary values, allowing instead positive rational number entries such that each column and row sums to $1$, then $X$ is a \textit{fractional isomorphism} \cite{ramana1994fractional}. The existence of a fractional isomorphism between two graphs is strictly weaker than the existence of an isomorphism, and so induces a less-refined equivalence relation on the class of all graphs. Fractional isomorphisms are an inherently \textit{linear algebraic} notion, and yet it has also been shown that two graphs are indistinguishable by the color-refinement algorithm if and only if a fractional isomorphism exists between them \cite{tinhofer1986graph,tinhofer1991note}.

The \textit{two-variable fragment} (\FOtwo) is the fragment of first-order logic in which only two variables are allowed. An important extension of \FOtwo is the \textit{two-variable fragment with counting quantifiers} (\Ctwo), which contains quantifiers of the form $\exists^{\geq k}$, where $\exists^{\geq k} x \varphi(x)$ asserts the existence of at least $k$ elements satisfying $\varphi(x)$. \Ctwo is an expressive, but decidable, fragment of \FOL \cite{gradel1997two}. A theorem of Cai et. al. shows that two graphs are \Ctwo-equivalent if and only if they are indistinguishable by the the color-refinement algorithm \cite{cai1992optimal}. In fact, they show that two graphs are invariant under the $k$-variable fragment with counting quantifiers (\Ck), which naturally generalizes \Ctwo, if and only if they are indistinguishable by the $(k-1)$-dimensional WL method (for $k \geq 2$).

In artificial intelligence, \textit{graph neural networks} (GNNs) are a type of machine learning architecture which have found numerous applications in the social and physical sciences \cite{wu2021comprehensive,zhou2020graph}. In \cite{morris2019weisfeiler}, Morris et. al. showed that GNNs can distinguish precisely those graphs distinguishable by the color-refinement algorithm. Inspired by the observation that \Ctwo and the color-refinement algorithm can be generalized to \Ck and the $k$-dimensional WL method, respectively, the authors proposed \emph{$k$-dimensional GNNs}. They showed that these $k$-dimensional GNNs can distinguish non-isomorphic graphs with the same expressive power as the $k$-dimensional WL method.

We have now seen that several seemingly distinct notions -- the color-refinement algorithm from graph theory, fractional isomorphism from linear algebra, the two-variable fragment with counting quantifiers from logic, and graph neural networks from machine learning -- all induce the same equivalence class on the class of undirected graphs. We have also seen that similar equivalences also hold for the natural generalizations of these notions. Furthermore, they are all undergirded by the same phenomenon: the expressive power of homomorphism count vectors restricted to particular classes of structures.

Due to these connections, Atserias et. al. set out to study which equivalence relations on graphs can be expressed by restricting homomorphism vectors to some fixed class of graphs \cite{atserias2021expressive}. In particular, they provide negative results showing that chromatic equivalence and $\FOk$-equivalence cannot be captured by homomorphism count indistinguishability with respect to any class of graphs. They also introduce a more general perspective, which we also take, in which ``counting'' can be performed in an arbitrary semiring.

\paragraph{Main Contributions.}
This paper aims to characterize logical equivalence with respect to various modal languages via homomorphism count indistinguishability with respect to appropriate classes of labeled transition systems (LTSs). The main results are as follows.
\begin{enumerate}
\item Positive-existential modal equivalence is captured by homomorphism count indistinguishability over the Boolean semiring with respect to the class of \emph{trees}. The extended languages with backward and global modalities are captured by the classes of \emph{connected, acyclic LTSs} and \emph{forests}, respectively, over the Boolean semiring.
\item Graded modal equivalence is captured by homomorphism count indistinguishability over the natural semiring with respect to the class of trees. The extended languages with backward and global modalities are captured by the classes of \emph{connected, acyclic} LTSs and \emph{forests}, respectively, over the natural semiring.
\item Equivalence with respect to hybrid logic is captured by homomorphism count indistinguishability over the natural semiring with respect to the class of \emph{point-generated} LTSs. The extended language with backward modalities is captured by the class of \emph{connected} LTSs.
\item Equivalence of LTSs with respect to positive modal logic and the basic modal language cannot be captured by restricting the left homomorphism count vector over any semiring to any class of LTSs.
\end{enumerate}

These results capture equivalence relations even over certain infinite structures, which we specify in the respective sections. The negative result (iv) is similar in spirit to the negative results from \cite{atserias2021expressive} mentioned above, but is more general in that it rules out homomorphism count indistinguishability characterizations for arbitrary semirings. Some of these results were obtained in the author's MSc thesis \cite{comer2023homomorphism}.

\section{Preliminaries}
\label{sec:prelim}
We assume familiarity with the syntax and semantics of first-order logic ($\FOL$). We use $\sigma$ and $\tau$ to denote first-order signatures, and we work primarily over \emph{modal signatures} of the form $\sigma = \Prop \cup \Act$, where $\Prop$ is a finite set of unary predicate symbols (called \emph{proposition letters}) and $\Act$ is a finite set of binary predicate symbols (called \emph{actions} or \emph{transitions}). All of the modal languages discussed in this paper will be variants of the basic (multi)modal language $\BML$, which is defined by the following recursive syntax:
$$\varphi := p \mid \varphi \land \varphi \mid \varphi \lor \varphi \mid \lnot \varphi \mid \Diamond_i \varphi \mid \square_i \varphi,$$
where $p \in \Prop$ and $\Diamond_i, \square_i$ are \emph{modalities} for the action $R_i \in \Act$. We define the semantics of $\BML$ by the well-known \emph{standard translation} of $\BML$ to $\FOL$:
\begin{multicols}{2}
\noindent
\begin{align*}
ST_x(p) &:= P(x), \\
ST_x(\varphi \land \psi) &:= ST_x(\varphi) \land ST_x(\psi), \\
ST_x(\varphi \lor \psi) &:= ST_x(\varphi) \lor ST_x(\psi),
\end{align*}
\columnbreak
\begin{align*}
ST_x(\lnot \varphi) &:= \lnot ST_x(\varphi), \\
ST_x(\Diamond_i \varphi) &:= \exists y (R_i(x,y) \land ST_y(\varphi)), \\
ST_x(\square_i \varphi) &:= \forall y (R_i(x,y) \to ST_y(\varphi)).
\end{align*}
\end{multicols}
\vspace*{-2em}
We write $M,N,S,T$ to denote (possibly infinite) first-order structures and $a,b,c,d,m,n,s,t$ to denote elements of structures. Given a relation symbol $F$, we write $F^M$ to denote the interpretation of $F$ in the structure $M$. For a $k$-ary relation symbol $F$, we say that $F^M(m_1,\hdots,m_k)$ \emph{holds} if the tuple $\langle m_1,\hdots,m_k \rangle$ is in $F^M$, in which case we say that $F^M(m_1,\hdots,m_k)$ is a \emph{fact} of $M$. Given a fact $f$, we write $\el(f)$ for the set of elements occurring in $f$. A \emph{pointed} structure, denoted $(M,a_1,\hdots,a_n)$, is a first-order structure $M$ together with a tuple of \emph{distinguished elements} $a_1,\hdots,a_n \in \dom(M)$. A \emph{labeled transition system} (LTS) is a pointed structure $M_a = (M,a)$ over a modal signature with exactly one distinguished element. We write that $M_a$ is a $\sigma$-LTS to emphasize that it is defined over the modal signature $\sigma$. We refer to elements of $\dom(M)$ as \emph{states}. Given a $\sigma$-LTS $M$ and a state $m \in \dom(M)$, we define
\begin{align*}
\Succ^M_\sigma[m] &:= \{n \in M \mid R^M(m,n) ~\text{holds for some}~ R \in \sigma \}, ~\text{and} \\
\Pred^M_\sigma[m] &:= \{n \in M \mid R^M(n,m) ~\text{holds for some}~ R \in \sigma \}
\end{align*}
to be the sets of $\sigma$-\emph{successors} and $\sigma$-\emph{predecessors}, respectively, of $m$ in $M$. For $R \in \sigma$, we also write $\Succ^M_R[m]$ and $\Pred^M_R[m]$ for the successors (resp. predecessors) of $m$ in $M$ along an $R$ transition. A $\sigma$-LTS $M_a$ is \emph{image-finite} if $\Succ^M_\sigma[m]$ is finite for each $m \in \dom(M)$, and \emph{degree-finite} if both $\Succ^M_\sigma[m]$ and $\Pred^M_\sigma[m]$ are finite for each $m \in \dom(M)$. We also write $\marking^M_\sigma(m)$ to denote the set of proposition letters $p \in \sigma$ such that $M_m \models p$. Each modal language $\logic$ discussed in this paper has an associated \emph{satisfaction relation} $\models$ between LTSs and formulas of $\logic$. If two $\sigma$-LTSs $M_a$ and $N_b$ satisfy the same formulas of $\logic$ over signature $\sigma$, then we write $M_a \equiv^\sigma_{\logic} N_b$.

\paragraph{Homomorphism Count Vectors.}
Let $(M,\tup{a})$ and $(N,\tup{b})$ be pointed $\sigma$-structures, where $\tup{a} = a_1,\hdots,a_n \in \dom(M)$ and $\tup{b} = b_1,\hdots,b_n \in \dom(N)$. A map $h: dom(M) \to dom(N)$ is a \emph{homomorphism} from $(M,\tup{a})$ to $(N,\tup{b})$ if $a_i \mapsto b_i$ for each $i \leq n$ and, for each $k$-ary relation symbol $R \in \sigma$, we have that $R^N(h(s_1),\hdots,h(s_k))$ holds whenever $R^M(s_1,\hdots,s_k)$ holds. An \emph{isomorphism} is a bijective homomorphism whose inverse is also a homomorphism; if an isomorphism from $(M,\tup{a})$ to $(N,\tup{b})$ exists, we write $(M,\tup{a}) \cong (N,\tup{b})$. A homomorphism $h: (M,\tup{a}) \to (N,\tup{b})$ is \emph{fully surjective} if it is surjective and if, for all $k$-ary relation symbols $R \in \sigma$, whenever $\langle t_1,\hdots,t_k \rangle \in R^N$, there also exists a tuple $\langle s_1,\hdots,s_k \rangle \in R^M$ such that $\langle h(s_1),\hdots,h(s_k) \rangle = \langle t_1,\hdots,t_k \rangle$. We say that $(M,\tup{a})$ and $(N,\tup{b})$ are \emph{homomorphically equivalent} if there exist homomorphisms $h: (M,\tup{a}) \to (N,\tup{b})$ and $g: (N,\tup{b}) \to (M,\tup{a})$.

Borrowing from database-theoretic terminology, an $\FOL$ formula of the form
$$\varphi(x_1,\hdots,x_n) := \exists y_1, \hdots, y_m \left( ~\bigwedge_{j \in J} \alpha_j~ \right),$$
where $J$ is a finite index set and each $\alpha_j$ is an atomic formula, is called a \emph{conjunctive query} (CQ). Each CQ $\varphi$ corresponds to a finite pointed structure whose domain is the variables of the formula, where each free variable is a distinguished element, and whose facts are the atomic formulas occurring in the formula \cite{chandra1977optimal}. This structure is the \emph{canonical instance} of $\varphi$ (notation: $\inst(\varphi)$). Any $\FOL$ formula containing only atomic formulas, existential quantifiers, and conjunction can be converted to a CQ by pulling all quantifiers to the front and renaming variables as necessary, so we will use the notation $\inst$ for arbitrary formulas of this form. The following useful fact equates satisfying assignments for a conjunctive query with homomorphisms out of its canonical instance.

\begin{fact}
\label{fact:magic-lemma}
Let $\varphi(x_1,\hdots,x_n)$ be a CQ and $(M,a_1,\hdots,a_n)$ a structure over the same signature. A homomorphism $h: \inst(\varphi) \to (M,a_1,\hdots,a_n)$ is a satisfying assignment for $\varphi$ in $(M,a_1,\hdots,a_n)$ such that $x_i \mapsto a_i$ for each $i \leq n$.
\end{fact}

We write $\Hom((M,\tup{a}),(N,\tup{a}))$ to denote the collection of all homomorphisms from $(M,\tup{a})$ to $(N,\tup{b})$. A \emph{semiring} is an algebraic structure $\mathcal{S} = \langle S, +, \cdot, 0, 1 \rangle$, where $\langle S, +, 0 \rangle$ is a commutative monoid, $\langle S, \cdot, 1 \rangle$ is a monoid, $\cdot$ distributes over $+$, and $a \cdot 0 = 0 \cdot a = 0$ for all $a \in A$. We define the \emph{homomorphism count} from $(M,\tup{a})$ to $(N,\tup{b})$ over $\mathcal{S}$ to be
$$\homS((M,\tup{a}),(N,\tup{a})) := \cnt_\mathcal{S}(\lvert \Hom((M,\tup{a}),(N,\tup{a})) \rvert),$$
where $\cnt_\mathcal{S}: \mathbb{N} \to S$ is defined by
$$\cnt_\mathcal{S}(n) := \begin{cases}
0_\mathcal{S} &\text{if}~ n = 0 \\
\sum_{1 \leq i \leq n} 1_\mathcal{S}, &\text{otherwise,}
\end{cases}$$
where the summation is defined by iterated addition in $\mathcal{S}$. Note that $\homS((M,\tup{a}),(N,\tup{a}))$ is only defined when $\lvert \Hom((M,\tup{a}),(N,\tup{a})) \rvert$ is finite. Our notion of counting is essentially just iterated addition, within some semiring, of the multiplicative unit of that semiring with itself. The decision to use semirings is not a canonical choice, but is general enough to cover all known results on homomorphism count indistinguishability.

\begin{definition}
\label{def:hom-count-vector}
Let $(M,\tup{a})$ be a $\tau$-structure with $n$ distinguished elements, and let $\mathcal{C}$ be a class of finite $\tau$-structures, each with $n$ distinguished elements, such that $\Hom((N,\tup{b}),(M,\tup{a}))$ is finite for each $(N,\tup{b})$ in $\mathcal{C}$. The \emph{left homomorphism vector} (or \emph{left profile}) of $(M,\tup{a})$ over $\mathcal{S}$ restricted to $\mathcal{C}$ is the $\mathcal{C}$-indexed sequence
$$\hom_\mathcal{S}(\mathcal{C},(M,\tup{a})) := \langle \hom_\mathcal{S}((N,\tup{b}),(M,\tup{a})) \rangle_{(N,\tup{b}) \in \mathcal{C}}.$$
The term \emph{left} is used here because the sequence includes homomorphism counts \emph{from} structures in $\mathcal{C}$ \emph{to} the structure $(M,\tup{a})$.
\end{definition}

We work mostly with the Boolean semiring $\mathbb{B} = \langle \{0,1\}, \lor, \land, \top, \bot \rangle$ and the natural number semiring $\mathbb{N} = \langle \omega, +, \cdot, 0, 1 \rangle$. We write $\TauModels$ for the class of all finite $\tau$-structures with $n$ distinguished elements. Note that $\homB((M,\tup{a}),(N,\tup{b})) = 1$ when a homomorphism from $(M,\tup{a})$ to $(N,\tup{b})$ exists, and $\homB((M,\tup{a}),(N,\tup{b})) = 0$ otherwise. It follows easily that $\homB(\TauModels,(M,\tup{a})) = \homB(\TauModels,(N,\tup{b}))$ if and only if $(M,\tup{a})$ and $(N,\tup{b})$ are homomorphically-equivalent. Using the notation of Definition \ref{def:hom-count-vector}, Lov\'{a}sz's theorem can be stated as follows.

\begin{theorem}[Lov\'{a}sz's Theorem, \cite{lovasz1967operations}]
Let $(M,\tup{a})$ and $(N,\tup{b})$ be finite $\tau$-structures with $n$ distinguished elements, where $\tau$ is a finite relational signature. Then $\homN(\TauModels,(M,\tup{a})) = \homN(\TauModels,(N,\tup{b}))$ if and only if $(M,\tup{a}) \cong (N,\tup{b})$.
\end{theorem}

\noindent The following definition was introduced in \cite{atserias2021expressive} to generalize Lov\'{a}sz's result.

\begin{definition}
\label{def:ext-class}
If $\mathcal{C}$ is a class of $\tau$-structures, we write $\Inj(\mathcal{C})$ to denote the class of $\tau$-structures $(N,\tup{b})$ such that there exists some injective homomorphism $h: (N,\tup{b}) \to (M,\tup{a})$ for some $(M,\tup{a}) \in \mathcal{C}$. Similarly, we write $\Sur(\mathcal{C})$ to denote the class of $\tau$-structures $(N,\tup{b})$ such that there exists some fully-surjective homomorphism $h: (M,\tup{a}) \to (N,\tup{b})$ for some $(M,\tup{a}) \in \mathcal{C}$. We define the \emph{extension class} of $\mathcal{C}$ to be $\Ext(\mathcal{C}) := \Inj(\mathcal{C}) \cap \Sur(\mathcal{C})$.
\end{definition}

\begin{theorem}[\cite{atserias2021expressive}]
\label{thm:AKW-lovasz}
Let $\mathcal{C}$ be a non-empty class of finite pointed $\tau$-structures, each with the same number of distinguished elements. Then for all $(M,\tup{a}), (N,\tup{b}) \in \mathcal{C}$, we have $\homN(\Ext(\mathcal{C}),(M,\tup{a})) = \homN(\Ext(\mathcal{C}),(N,\tup{b}))$ if and only if $(M,\tup{a}) \cong (N,\tup{b})$.
\end{theorem}

\paragraph{Important Classes of Structures.}
We now define the classes of structures relevant to our results (examples of each can be found in Figure \ref{fig:lts-examples}). Let $M_a$ be a $\sigma$-LTS. Given states $m, n \in \dom(M)$, a \emph{$\sigma$-path of length $k$ from $m$ to $n$} is a sequence $\pi = \langle f_1,\hdots,f_k \rangle$ of binary facts such that $m \in \el(f_1)$, $n \in \el(f_k)$, and $\el(f_i) \cap \el(f_{i+1}) \neq \emptyset$ for each $i < k$. A $\sigma$-path is \emph{simple} if it contains no duplicate facts. A \emph{connected component} of $M_a$ is a maximal set $S \subseteq \dom(M)$ such that, for each distinct pair of states $m,n \in S$, there exists a $\sigma$-path $\pi$ from $m$ to $n$. We say that $M_a$ is \emph{connected} if $\dom(M)$ is a connected component of $M_a$, and we say that $M_a$ is \emph{acyclic} if there are no simple $\sigma$-paths from some $m \in \dom(M)$ to itself. A \emph{directed $\sigma$-path of length $k$} from $m$ to $n$ is a length-$k$ tuple $\langle (b_0,b_1), (b_1,b_2), \hdots, (b_{k-1},b_k) \rangle$ such that for each $j < k$, there is some $R \in \sigma$ such that $R^M(b_j,b_{j+1})$ holds. Note that all directed $\sigma$-paths can be seen as a special case of $\sigma$-paths.

A $\sigma$-LTS $M_a$ is \emph{point-generated} if, for each $m \in \dom(M)$, there's a directed $\sigma$-path from $a$ to $m$. If there is a unique directed $\sigma$-path from $a$ to each $m \in \dom(M)$, then $M_a$ is a \emph{$\sigma$-tree}. The \emph{depth} of a state $m$ in a point-generated $\sigma$-LTS $M_a$ is the length $\depth(m)$ of the shortest directed $\sigma$-path from $a$ to $m$; we set $\depth(a) = 0$. The \emph{depth} of a point-generated $\sigma$-LTS is the supremum of the depths of its elements. Given a point-generated $\sigma$-LTS $M_a$, we define $M^k_a$ to be the pointed substructure of $M_a$ containing all elements in $\dom(M)$ of depth at most $k$. Depth for connected structures is defined analogously via $\sigma$-paths. If $(M_j,a_j)$ is a $\sigma$-tree for each $j \in J$, where $J$ is some finite index set, and $(M,a)$ is a $\sigma$-LTS obtained by taking the disjoint union $M = \biguplus_{j \in J} M_j$ and setting $a = a_j$ for some $j \in J$, then $(M,a)$ is a \emph{$\sigma$-forest}.

\begin{definition}
\label{def:classes}
We use the following notation for these classes of structures.
\begin{enumerate}
\item $\KT^k_\sigma$ is the class of finite $\sigma$-trees of depth at most $k$.
\item $\AKT^k_\sigma$ is the class of finite connected, acyclic $\sigma$-LTSs of depth at most $k$.
\item $\FKT_\sigma$ is the class of finite $\sigma$-forests.
\item $\GK^k$ is the class of finite point-generated $\sigma$-LTSs of depth at most $k$.
\item $\C^k_\sigma$ is the class of finite connected $\sigma$-LTSs of depth at most $k$.
\item We set $\KT_\sigma := \bigcup_{k \in \omega} \KT^k_\sigma$; we define $\AKT_\sigma$, $\PG_\sigma$, and $\C_\sigma$ similarly.
\end{enumerate}
\end{definition}

\begin{figure}
\centering
\begin{subfigure}[b]{0.20\textwidth}
\centering
\begin{tikzpicture}
\node (1) at (0,0) {$\bullet$};
\node (2) at (1,0) {$\bullet$};
\node (3) at (0.5,1) {$a$};
\node (4) at (-0.4,-1) {$\bullet$};
\node (5) at (0.4,-1) {$\bullet$};
\node (6) at (1.4,-1) {$\bullet$};
\draw[->] (3) -- (1);
\draw[->] (3) -- (2);
\draw[->] (1) -- (4);
\draw[->] (1) -- (5);
\draw[->] (2) -- (6);
\end{tikzpicture}
\caption{A $\sigma$-tree}
\end{subfigure}
\hspace*{25pt}
\begin{subfigure}[b]{0.20\textwidth}
\centering
\begin{tikzpicture}
\node (1) at (0,0) {$\bullet$};
\node (2) at (1,0) {$\bullet$};
\node (3) at (0.5,1) {$a$};
\node (4) at (-0.4,-1) {$\bullet$};
\node (5) at (0.4,-1) {$\bullet$};
\node (6) at (1.4,-1) {$\bullet$};
\draw[->] (3) -- (1);
\draw[->] (2) -- (3);
\draw[->] (1) -- (4);
\draw[->] (5) -- (1);
\draw[->] (2) -- (6);
\end{tikzpicture}
\caption{A connected, acyclic $\sigma$-LTS}
\end{subfigure}
\hspace*{25pt}
\begin{subfigure}[b]{0.285\textwidth}
\centering
\begin{tikzpicture}
\node (1) at (0,0) {$\bullet$};
\node (2) at (1,0) {$\bullet$};
\node (3) at (0.5,1) {$a$};
\node (4) at (-0.4,-1) {$\bullet$};
\node (5) at (0.4,-1) {$\bullet$};
\node (6) at (1.4,-1) {$\bullet$};
\draw[->] (3) -- (1);
\draw[->] (3) -- (2);
\draw[->] (2) edge [loop above, out=-20, in=75, looseness = 5] (2);
\draw[->] (1) -- (4);
\draw[->] (1) -- (5);
\draw[->] (2) -- (6);
\draw[->] (2) -- (5);
\draw[->] (4) edge [loop above, out=-50, in=-130, looseness = 5] (4);
\end{tikzpicture}
\caption{A point-generated $\sigma$-LTS}
\end{subfigure}
    
\vspace{0.5cm}
\begin{subfigure}[b]{0.35\textwidth}
\centering
\begin{tikzpicture}
\node (1) at (0,0) {$\bullet$};
\node (2) at (1,0) {$a$};
\node (3) at (0.5,1) {$\bullet$};
\node (4) at (-0.4,-1) {$\bullet$};
\node (5) at (0.4,-1) {$\bullet$};
\node (6) at (1.4,-1) {$\bullet$};
\draw[->] (3) -- (1);
\draw[->] (2) -- (3);
\draw[->] (2) edge [loop above, out=-20, in=75, looseness = 5] (2);
\draw[->] (1) -- (4);
\draw[->] (5) -- (1);
\draw[->] (6) -- (2);
\draw[->] (2) -- (5);
\draw[->] (4) edge [out=130, in=170] (1);
\end{tikzpicture}
\caption{A connected $\sigma$-LTS}
\end{subfigure}
\hspace*{20pt}
\begin{subfigure}[b]{0.35\textwidth}
\centering
\begin{tikzpicture}
\node (1) at (0,0) {$\bullet$};
\node (2) at (1,0) {$\bullet$};
\node (3) at (0.5,1) {$a$};
\node (4) at (-0.4,-1) {$\bullet$};
\node (5) at (0.4,-1) {$\bullet$};
\node (6) at (1.4,-1) {$\bullet$};
\draw[->] (3) -- (1);
\draw[->] (3) -- (2);
\draw[->] (1) -- (4);
\draw[->] (1) -- (5);
\draw[->] (2) -- (6);
\node (7) at (2.6,1) {$\bullet$};
\node (8) at (2.6,0) {$\bullet$};
\node (9) at (2.2,-1) {$\bullet$};
\node (10) at (3,-1) {$\bullet$};
\draw[->] (7) -- (8);
\draw[->] (8) -- (9);
\draw[->] (8) -- (10);
\end{tikzpicture}
\caption{A $\sigma$-forest}
\end{subfigure}
\caption{Examples of $\sigma$-LTSs.}
\label{fig:lts-examples}
\end{figure}

\noindent The following class inclusions are clear from the definitions:
$$\KT_\sigma \subseteq \PG_\sigma \subseteq \C_\sigma, \quad \KT_\sigma \subseteq \AKT_\sigma \subseteq \C_\sigma, \quad \text{and} \quad \KT_\sigma \subseteq \FKT_\sigma.$$
The following two facts are easily verified (cf. Definition \ref{def:ext-class}).

\begin{fact}
\label{prop:tree-ext-equivalence}
$\KT^k_\sigma = \Ext(\KT^k_\sigma)$.
\end{fact}

\begin{fact}
\label{prop:pg-ext-equivalence}
$\PG^k_\sigma = \Ext(\PG^k_\sigma)$.
\end{fact}

\noindent The next lemma, used in Sections \ref{sec:graded-ml} and \ref{sec:hybrid-logic}, is proven by constructing an ascending chain of local isomorphisms whose union is a full isomorphism.

\begin{restatable}{lemma}{isoext}
\label{lemma:iso-ext}
If $M_a$ and $N_b$ are point-generated $\sigma$-LTSs such that $M^k_a$ and $N^k_b$ are finite and isomorphic for all $k \in \mathbb{N}$, then $M_a \cong N_b$.
\end{restatable}

For the remainder of the paper, we fix a modal signature $\sigma = \Prop \cup \Act$, where $\Prop$ is a finite set of (unary) proposition letters and $\Act = \{ R_i \mid i \in I \}$ is a set of (binary) \emph{actions} (or \emph{transitions}) indexed by some finite set $I$.

\section{Positive-Existential Modal Logic}
\label{sec:pos-ext-ml}
We begin with a characterization of equivalence with respect to positive-existential modal logic (notation: $\PML$) by restricting the left homomorphism vector over the Boolean semiring to the class of $\sigma$-trees. $\PML$ is the fragment of $\BML$ lacking both negation and the $\square$ modality, and we write $\PMLk$ for the collection of $\PML$ formulas of modal depth at most $k$. The key observation leading to this theorem is the following proposition.

\begin{proposition}
\label{prop:trees-pml-equiv}
A $\sigma$-LTS $T_c$ is in $\KT^k_\sigma$ if and only if $T_c \cong \inst(ST_x(\varphi))$ for some disjunction-free $\varphi \in \PMLk$, where $ST_x$ denotes the standard translation.
\end{proposition}
\begin{proof}
For the forward direction, we show by induction on the depth of $\sigma$-trees $T_c$ that $T_c$ is isomorphic to the canonical instance of some $\PMLk$ formula. For each element $s \in T$, define $\textrm{mark}^{+,T}_s:= \bigwedge_{p \in \marking^T_\sigma(s)} p$. For the base case, if $\depth(T_c) = 0$, then $\dom(T) = \{ c \}$. Then clearly $\inst(ST_x(\textrm{mark}^{+,T}_c)) \cong T_c$. Now suppose that every $\sigma$-tree of depth $j < k$ is isomorphic to the canonical instance of some $\PMLk$ formula, and let $T_c$ be an arbitrary $\sigma$-tree of depth $k$. Let $\Succ^T_\sigma[c] = \{ s_1,\hdots,s_n \}$, and let $T^1_{s_1},\hdots,T^n_{s_n}$ denote the corresponding rooted subtrees of $T_c$. By the inductive hypothesis, there exist formulas $\varphi_1,\hdots,\varphi_n$ such that $\inst(ST_x(\varphi_i)) \cong T^i_{s_i}$ for each $i \leq n$. For each $i \leq n$, let $j_i$ be the unique index in $I$ such that $R^T_{j_i}(c,s_i)$ holds. Then $\inst(ST_x(\textrm{mark}^{+,T}_s \land \bigwedge_{i \leq n} \Diamond_{j_i} \varphi_i))$ is easily seen to be isomorphic to $T_c$. 

For the reverse direction, we show by induction on the complexity of $\PML$ formulas $\varphi$ that $\inst(ST_x(\varphi))$ is a $\sigma$-tree. For the base case, if $\varphi = p$ for some $p \in \Prop$, then $\inst(ST_x(\varphi))$ is a single state at which the proposition letter $p$ is true, which is a $\sigma$-tree. For the inductive step, either $ST_x(\varphi) = ST_x(\psi_1) \land ST_x(\psi_2)$ for some formulas $\psi_1,\psi_2$, or $ST_x(\varphi) = \exists y (R_i(x,y) \land \psi)$ for some formula $\psi$. In the first case, $\inst(ST_x(\varphi))$ is the $\sigma$-tree obtained by equating the roots of $\inst(ST_x(\psi_1))$ and $\inst(ST_x(\psi_2))$. In the second case, $\inst(ST_x(\varphi))$ is the $\sigma$-tree obtained by adding a new root to $\inst(ST_x(\psi))$, with the old root as its unique $R_i$-successor.
\end{proof}

Note that if $T_c$ is a finite $\sigma$-tree and $M_a$ is an image-finite $\sigma$-LTS, then there are only finitely many homomorphisms from $T_c$ to $M_a$. It follows that if $M_a$ is image-finite, then $\homS(\KT^k,M_a)$ is well-defined for all semirings $\mathcal{S}$.

\begin{theorem}
\label{thm:sim-eq}
If $M_a$ and $N_b$ are image-finite $\sigma$-LTSs, then $M_a \PMLkequiv N_b$ if and only if $\homB(\KT^k_\sigma,M_a) = \homB(\KT^k_\sigma,N_b)$.
\end{theorem}
\begin{proof}
For the left-to-right direction, suppose that $M_a \PMLkequiv N_b$, and let $T_c$ be an arbitrary finite $\sigma$-tree of depth at most $k$. By Proposition \ref{prop:trees-pml-equiv}, let $\varphi$ be a disjunction-free $\PMLk$ formula such that $T_c \cong \inst(ST_x(\varphi))$. Then
\begin{align*}
\homB(T_c,M_a) = 1 &\iff M_a \models \varphi \hspace*{40pt} &\text{(Fact \ref{fact:magic-lemma})} \\
&\iff N_b \models \varphi &\text{(Assumption)} \\
&\iff \homB(T_c,N_b) = 1 &\text{(Fact \ref{fact:magic-lemma})}.
\end{align*}
Hence $\homB(\KT^k_\sigma,M_a) = \homB(\KT^k_\sigma,N_b)$. The other direction is symmetric.
\end{proof}

\begin{corollary}
If $M_a$ and $N_b$ are image-finite $\sigma$-LTSs, then $M_a \PMLequiv N_b$ if and only if $\homB(\KT_\sigma,M_a) = \homB(\KT_\sigma,N_b)$.
\end{corollary}

\paragraph*{$\PML$ with backward and global modalities.} 
We now state two results for $\PML$ extended with backward and global modalities. The proofs are similar to, yet simpler than, those for $\GML$ with the backward and global modalities given in Section 4, and so we omit them.

\begin{definition}
\label{def:backward-modalities}
Given a pointed $\sigma$-LTS $M_a$ and a formula $\varphi$, we define
\begin{center}
\begin{tabular}{l l l}
$M,a \models \dback_i^{\geq k} \varphi$ &if there exist at least $k$ many elements \\
&$b \in \Pred^M_{R_i}[a]$ such that $M,b \models \varphi$.
\end{tabular}
\end{center}
\end{definition}

We call $\dback^{\geq k}_i$ a \emph{backward modality} for the action $R_i$, where $i \in I$. We write $\PMLb$ for the extension of $\PML$ with the modalities $\dback_i^{\geq 1}$, and $\PMLbk$ for the fragment of $\PMLb$ containing formulas of modal depth at most $k$.

\begin{theorem}
\label{thm:back-sim}
If $M_a$ and $N_b$ are degree-finite $\sigma$-LTSs, then $M_a \PMLbkequiv^\sigma N_b$ if and only if $\homB(\AKT^k_\sigma,M_a) = \homB(\AKT^k_\sigma,N_b)$.
\end{theorem}

\begin{definition}
\label{def:global-modality}
Given a pointed $\sigma$-LTS $M_a$ and a formula $\varphi$, we define
\begin{center}
\begin{tabular}{l l l}
$M_a \models \glob \varphi$ &if there exist at least $k$ many elements \\
&$b \in M$ such that $M_b \models \varphi$.
\end{tabular}
\end{center}
\end{definition}

\noindent We refer to $\glob$ as a \emph{global modality}. Let $\PMLg$ denote the extension of $\PML$ with the global modality for $k=1$.

\begin{theorem}
\label{thm:global-sim}
If $M_a$ and $N_b$ are finite $\sigma$-LTSs, then $M_a \PMLgequiv^\sigma N_b$ if and only if $\homB(\FKT_\sigma,M_a) = \homB(\FKT_\sigma,N_b)$.
\end{theorem}

Note that Theorems \ref{thm:sim-eq}, \ref{thm:back-sim}, and \ref{thm:global-sim} are stated for image-finite, degree-finite, and finite LTSs, respectively, \emph{only} due to the fact that our notion of counting requires a finite number of homomorphisms in order to be well-defined. However, this is an artificial constraint: if we were to treat the Boolean homomorphism count only as an indicator that $\Hom(T_c,M_a)$ is non-empty, then the Boolean left profiles with respect to $\KT_\sigma$, $\AKT_\sigma$, and $\FKT_\sigma$ would be well-defined for, and hence all of these results would apply to, \emph{arbitrary} $\sigma$-LTSs.

\section{Graded Modal Logic}
\label{sec:graded-ml}
We now turn to graded modal logic (notation: $\GML$), which is the extension of the basic modal language with \emph{graded modalities} $\Diamond_i^{\geq k}$ for each $k \in \mathbb{Z}^+$ and $i \in I$, where $\Diamond_i^{\geq k} \varphi$ asserts that there are at least $k$ many $R_i$-successors of the current state at which $\varphi$ is true \cite{fine1972in,goble1970grades}. Recall the notion of a tree-unraveling.

\begin{definition}
The \textit{unraveling} of a $\sigma$-LTS $M_a$ is the $\sigma$-LTS $\unr(M_a)$ where
\begin{enumerate}
\item $\dom(\unr(M_a))$ is the set of strings $w = \langle w_1,\hdots,w_n \rangle$ over $\dom(M)$ with $a = w_1$ and for each $i < n$, there is $R \in \Act$ such that $R^M(w_i,w_{i+1})$ holds,
\item $R^{\unr(M_a)} = \{ (w,w^{\frown}\langle u \rangle) \mid (\last(w),u) \in R^M \}$ for each $R \in \Act$, and
\item $p^{\unr(M_a)} = \{ w \in \dom(\unr(M_a)) \mid \last(w) \in p^M \}$ for $p \in \Prop$,
\end{enumerate}
where $\langle a \rangle$ is the unique distinguished element of the model, $\last$ is the function mapping strings to their last element, and $w^{\frown}w'$ denotes string concatenation.
\end{definition}

If $M_a$ is a $\sigma$-LTS, then $\unr(M_a)$ is a (possibly infinite) $\sigma$-tree. Furthermore, if $M_a$ is image-finite, then $\unr^k(M_a)$, the substructure of $\unr(M_a)$ containing only states of depth at most $k$, is a finite $\sigma$-tree of depth $k$. The unraveling construction is known to preserve the truth of $\GML$ formulas.

\begin{theorem}[Unraveling Invariance, \cite{van2009lindstrom}]
\label{thm:unr-inv}
If $M_a$ is a $\sigma$-LTS, then for each $\varphi \in \GML$, we have that $M_a \models \varphi$ if and only if $\unr(M_a) \models \varphi$.
\end{theorem}

\noindent In fact, $\GMLk$ formulas can describe finite $\sigma$-trees of depth $k$ up to isomorphism. Recall that if $M_a$ is a $\sigma$-tree, then $M^k_a$ denotes the substructure of $M_a$ containing only elements of depth at most $k$. We write $\unr^k(M_a)$ for the substructure of $\unr(M_a)$ containing only elements of depth at most $k$.

\begin{proposition}
\label{prop:finite-tree-iso}
For each $T_c \in \KT^k_\sigma$, there is a formula $\varphi \in \GMLk$ such that, if $M_a$ is a $\sigma$-tree, then $M_a \models \varphi$ if and only if $M^k_a \cong T_c$.
\end{proposition}
\begin{proof}
For any state $s \in T$, define $\markform^T_s:= \left( \bigwedge_{p \in \marking^T_\sigma(s)} p \right) \land \left( \bigwedge_{p \not \in \marking^T_\sigma(s)} \lnot p \right)$. We proceed by strong induction on $k$. For $T_c \in \KT^0_\sigma$, clearly $\markform^T_s$ meets the requirements of the claim. Now suppose the claim holds for $\KT^j_\sigma$ for all $j < k$, and let $T_c \in \KT^k_\sigma$. Let $\Succ^T_\sigma[c] = \{s_1,\hdots,s_n\}$, and let $T^i_{s_i}$ denote the subtree of $T_c$ rooted at $s_i$ for each $i \leq n$. By the inductive hypothesis, there is a formula $\varphi_i$ satisfying the claim for each $T^i_{s_i}$. Some of these may be equivalent, so let $\varphi'_1,\hdots,\varphi'_m$ be the sequence of formulas obtained by removing duplicates. For each $i \leq m$, let $n_i$ denote the number of elements $s_j \in \Succ^T_\sigma[c]$ such that the formula associated with $T^j_{s_j}$ is $\varphi'_i$. The formula $\varphi := \left( \markform^T_c \land \Diamond_i^{n} \top \land \bigwedge_{i \leq m} \Diamond_i^{=n_i} \varphi'_i \right)$ satisfies the requirements of the claim.
\end{proof}

\noindent The next lemma shows that homomorphism counts from finite $\sigma$-trees are preserved between a $\sigma$-LTS and its unraveling up to depth $k$.

\begin{lemma}
\label{thm:unr-hom-count-preservation}
If $M_a$ is an image-finite LTS, then for each $k \in \mathbb{N}$, we have that $\homN(\KT^k,M_a) = \homN(\KT^k,\unr^k(M_a))$.
\end{lemma}
\begin{proof}
Given a directed tree-shaped LTS $T_c$ of depth at most $k$, we construct injections between $\Hom(T_c,M_a)$ and $\Hom(T_c,\unr^k(M_a))$.

($\leq$) For each $h \in \Hom(T_c,M_a)$, we define partial maps $\hat{h}_i: T_c \to \unr^k(M_a)$ for $i \leq k$ by recursion on the depth of elements of $T_c$, where $\hat{h}_0(c) = \langle a \rangle$ and $\hat{h}_{i+1}(m) = \hat{h}_i(\parent(m))^{\frown}h(m)$.
We claim that $\hat{h} = \bigcup_{i \leq k} \hat{h}_i$ is a homomorphism. To see that $\hat{h}$ preserves proposition letters, observe that, for any proposition letter $p \in P$, we have that if $p^T(m)$ holds, then $p^M(h(m))$ holds since $h$ is a homomorphism, and so $p^{\unr^k(M_a)}(\hat{h}(m))$ holds since $h(m) = \last(\hat{h}(m))$.

To show that $\hat{h}$ preserves actions, it suffices to show that, for all $R \in \Act$, all states $m \in T_c$, and all $s \in \Succ^M_R[m]$, we have that $R^{\unr^k(M_a)}(\hat{h}(m),s)$ holds. This follows from the observations that (1) $\hat{h}(s)$ extends $\hat{h}(m)$ (by the definition of $\hat{h}$), and (2) if $R^M(\last(\hat{h}(x)), \last(\hat{h}(s)))$ holds, then $R^M(h(x),h(s_j))$ (since $h_i$ is a homomorphism), in which case $R^{\unr^k(M_a)}(\hat{h}(m),s)$ holds by the definition of unravelings. Thus $\hat{h}$ is a homomorphism. Furthermore, it's clear that the map $h \mapsto \hat{h}$ is an injection from $\Hom(T_c,M_a)$ to $\Hom(T_c,\unr^k(M_a))$.

($\geq$) For each $g \in \Hom(T_c,\unr^k(M_a))$, define $\hat{g}: T_c \to M_a$ to be the map $m \mapsto \last(g(m))$. By the definition of unravelings, $\hat{g}$ is a homomorphism. We claim that $g \mapsto \hat{g}$ is an injective map from $\Hom(T_c,\unr^k(M_a))$ to $\Hom(T_c,M_a)$. To see this, let $g,g': T_c \to unr(M_a)$ be homomorphisms, and let $\hat{g},\hat{g}'$ be the corresponding maps in $\Hom(T_c,M_a)$. Suppose that $\hat{g} = \hat{g}'$. We now show by induction on depth of the elements of $T_c$ that $g = g'$.

The base case is immediate, since $g(c) = g'(c) = \langle a \rangle$. Now suppose inductively that $g$ and $g'$ agree on all elements of depth less than $k$, and let $m \in T_c$ be some element of depth $k$. By assumption, we have that $\hat{g}(m) = \hat{g}'(m)$, and so $last(g(m)) = last(g'(m))$. Let $n$ denote the unique predecessor of $m$ (i.e., its parent). Clearly $n$ has depth less than $k$, and so $g(n) = g'(n)$. Since $g$ and $g'$ are homomorphisms and $R_i(n,m)$ holds for some $i \in I$, we have that $R^{unr(M_a)}_i(g(n),g(m))$ and $R^{unr(M_a)}_i(g'(n),g'(m))$ hold. Then by the definition of the actions for unravelings, we have that $g(m) = g(n)^{\frown} \hat{g}(m) = g'(m)$.
\end{proof}

\noindent We are now ready to prove our characterization result for $\GML$.

\begin{theorem}
\label{thm:gml-main}
For image-finite LTSs $M_a$ and $N_b$, the following are equivalent:
\begin{enumerate}
\item $\homN(\KT^k,M_a) = \homN(\KT^k,N_b)$,
\item $\unr^k(M_a) \cong \unr^k(N_b)$,
\item $M_a \GMLkequiv^\sigma N_b$.
\end{enumerate}
\end{theorem}
\begin{proof}
The equivalence of $(ii)$ and $(iii)$ is a consequence of Proposition \ref{prop:finite-tree-iso}, Theorem \ref{thm:unr-inv} and the observation that satisfaction of $\GMLk$ formulas in $\sigma$-trees depends only on the elements up to depth $k$. For $(i)$ to $(ii)$, suppose that $\homN(\KT^k,M_a) = \homN(\KT^k,N_b)$. Then by Lemma \ref{thm:unr-hom-count-preservation}, we have that $\homN(\KT^k,\unr^k(M_a)) = \homN(\KT^k,\unr^k(N_b))$. Since $\unr^k(M_a),\unr^k(N_b) \in \KT^k$, this implies, by Fact \ref{prop:tree-ext-equivalence} and Theorem \ref{thm:AKW-lovasz}, that $\unr^k(M_a) \cong \unr^k(N_b)$. For $(ii)$ to $(i)$, suppose that $\unr^k(M_a) \cong \unr^k(N_b)$. Then $\homN(\KT^k,\unr^k(M_a)) = \homN(\KT^k,\unr^k(N_b))$, and so $\homN(\KT^k,M_a) = \homN(\KT^k,N_b)$ by Lemma \ref{thm:unr-hom-count-preservation}.
\end{proof}

\noindent Using Lemma \ref{lemma:iso-ext}, we easily obtain the following corollary.

\begin{corollary}
\label{cor:gml-main}
For image-finite LTSs $M_a$ and $N_b$, the following are equivalent:
\begin{enumerate}
\item $\homN(\KT,M_a) = \homN(\KT,N_b)$,
\item $\unr(M_a) \cong \unr(N_b)$,
\item $M_a \GMLequiv^\sigma N_b$.
\end{enumerate}
\end{corollary}

\paragraph*{$\GML$ with backward modalities.}
Let $\GMLb$ denote the extension of $\GML$ with backward modalities for each $k \in \mathbb{N}$ (cf. Definition \ref{def:backward-modalities}). We write $\GMLbk$ for the fragment of $\GMLb$ formulas of modal depth at most $k$. Fix an expansion $\sigma_B = \Prop \cup \Act \cup \Act_B$ of $\sigma$, where $\Act_B = \{ B_i \mid i \in I \}$ is disjoint from $\Act$.

\begin{definition}
\label{def:backward-expansion}
The \emph{backward expansion} of a $\sigma$-LTS $M_a$ is the $\sigma_B$-expansion $M^B_a$ of $M_a$ given by setting $B_i^{M^B} = \{ \langle n,m \rangle \mid R_i^M(m,n) ~\text{holds} \}$ for each $i \in I$.
\end{definition}

\noindent Recall that $\AKT^k_\sigma$ denotes the class of connected acyclic $\sigma$-LTSs of depth at most $k$, and that $\AKT_\sigma = \bigcup_{k \in \omega} \AKT^k_\sigma$ (cf. Definition \ref{def:classes}).

\begin{definition}
\label{def:dirtree-around-Tc}
Given some $T_c \in \AKT_\sigma$, we define a $\sigma_B$-LTS $T^\downarrow_c := (T^\downarrow,c)$ with $\dom(T^\downarrow) := \dom(T)$ and $\marking^{T^\downarrow}_\sigma(m) := \marking^T_\sigma(m)$ for all $m \in \dom(T)$, where
\begin{enumerate}
\item If $R^T_i(m,n)$ holds where $depth(m) < depth(n)$, then $R^{T^\downarrow}_i(m,n)$ holds.
\item If $R^T_i(m,n)$ holds where $depth(n) < depth(m)$, then $B^{T^\downarrow}_i(m,n)$ holds.
\end{enumerate}
Intuitively, $(\cdot)^\downarrow$ replaces all $R_i$ transitions ``pointing toward'' the root $c$ with $B_i$ transitions in the opposite direction. For all $T_c \in \AKT_\sigma$, clearly $T^\downarrow_c$ is a $\sigma_B$-tree.
\end{definition}

\begin{definition}
\label{def:flip}
Let $S_d$ be a $\sigma_B$-LTS. Define a $\sigma$-LTS $\flip(S_d) := (\flip(S),d)$ with $\dom(\flip(S)) = \dom(S)$ and $\marking^{\flip(S)}_\sigma(m) := \marking^S_\sigma(m)$ for all $m \in \dom(S)$, where $R_i^{\flip(S)} = R_i^{S} \cup (B_i^S)^{-1}$ for each $i \in I$.
\end{definition}

Intuitively, $\flip$ forms a $\sigma$-LTS from a $\sigma_B$-LTS $S_d$ by replacing $B_i$ transitions in $S_d$ by the corresponding $R_i$ transition in the opposite direction. If $S_d$ is a $\sigma_B$-tree, then $\flip(S_d)$ is a connected acyclic $\sigma$-LTS. The $(\cdot)^\downarrow$ transformation on connected, acyclic $\sigma$-LTSs and the $\flip$ transformation on $\sigma_B$-trees are exact inverses of one another: $T_c = \flip(T^\downarrow_c)$ for all connected acyclic $\sigma$-LTSs $T_c$, and $S_d = (\flip(S_d))^\downarrow$ for all $\sigma_B$-trees $S_d$. These operations also clearly preserve the depth of the structures to which they are applied.

When we consider homomorphisms from finite connected acyclic $\sigma$-LTSs $T_c$, image-finiteness is not enough to guarantee that $\Hom(T_c,M_a)$ is finite. However, if $M_a$ is degree-finite, then $\bigcup_{h \in Hom(T_c,M_a)} Im(h)$
is finite, and hence $\homN(\AKT_\sigma,M_a)$ is well-defined. Furthermore, note that if $M_a$ is a degree-finite $\sigma$-LTS, then $M^B_a$ is also a degree-finite $\sigma^B$-LTS.

\begin{proposition}
\label{prop:gml-back-hom-count-pres}
Let $M_a$ be a degree-finite $\sigma$-LTS. Then
\begin{enumerate}
\item If $T_c$ is in $\AKT_\sigma$, then $\homN(T_c,M_a) = \homN(T^\downarrow_c,M^B_a)$.
\item If $T_c$ is in $\KT_{\sigma_B}$, then $\homN(T_c,M^B_a) = \homN(\flip(T_c),M_a)$.
\end{enumerate}
\end{proposition}
\begin{proof}[Sketch]
For part $(i)$, we show that a map $h: \dom(T) \to \dom(M)$ is a homomorphism from $T_c$ to $M_a$ if and only if it is also a homomorphism from $T^\downarrow_c$ to $M^B_a$, which is straightforward from the definitions of $(\cdot)^\downarrow$ and $M^B_a$. The proof of part $(ii)$ is analogous.
\end{proof}

\begin{lemma}
\label{lemma:gml-back-count-equiv}
Let $M_a$ and $N_b$ be degree-finite $\sigma$-LTSs. Then $$\homN(\AKT^k_\sigma,M_a) = \homN(\AKT^k_\sigma,N_b) \iff \homN(\KT^k_{\sigma_B},M^B_a) = \homN(\KT^k_{\sigma_B},N^B_b).$$
\end{lemma}
\begin{proof}[Sketch]
Observe that $(\cdot)^\downarrow$ is a bijective map from $\AKT^k_\sigma$ to $\KT^k_{\sigma_B}$, while $\flip$ is its inverse. The forward direction is by contraposition. If we have $\homN(T_c,M^B_a) \neq \homN(T_c,N^B_b)$ for some $\sigma_B$-tree $T_c$ of depth at most $k$, then we have $\homN(T^\downarrow_c,M_a) \neq \homN(T^\downarrow_c,N_b)$ by Proposition \ref{prop:gml-back-hom-count-pres}. The reverse direction is proven by contraposition in a similar fashion.
\end{proof}

\begin{lemma}
\label{lemma:gml-back-logic-equiv}
Let $M_a$ and $N_b$ be degree-finite $\sigma$-LTSs. Then $M_a \GMLbkequiv^\sigma N_b$ if and only if $M^B_a \GMLkequiv^{\sigma_B} N^B_b$.
\end{lemma}
\begin{proof}[Sketch]
Consider the translation $\tr$ from $\GMLbk$ to $\GMLk$ which replaces backward modalities with the corresponding forward modalities in $\Act_{\sigma_B}$. It is a straightforward induction to show that $M_a \models \varphi$ if and only if $M^B_a \models \tr(\varphi)$. Since this translation is bijective, the result follows immediately.
\end{proof}

\noindent
We now prove our $\GMLb$ characterization result.

\begin{theorem}
\label{thm:gml-back-result}
Let $M_a$ and $N_b$ be degree-finite $\sigma$-LTSs. Then
$$\hom(\AKT^k_\sigma,M_a) = \hom(\AKT^k_\sigma,N_b) \iff M_a \GMLbkequiv^\sigma N_b.$$
\end{theorem}
\begin{proof}
Let $M_a$ and $N_b$ be degree-finite $\sigma$-LTSs. Then we have that
\begin{align*}
\hom(\AKT^k_\sigma,M_a) &= \hom(\AKT^k_\sigma,N_b) \\
&\iff \hom(\KT^k_{\sigma_B},M^B_a) = \hom(\KT^k_{\sigma_B},N^B_b) &(\text{Lemma}~ \ref{lemma:gml-back-count-equiv}) \\
&\iff M^B_a \GMLkequiv^{\sigma_B} N^B_b &(\text{Theorem}~ \ref{thm:gml-main}) \\
&\iff M_a \GMLbkequiv^\sigma N_b. &(\text{Lemma}~ \ref{lemma:gml-back-logic-equiv})
\end{align*}
This completes the proof.
\end{proof}

\paragraph*{$\GML$ with the global modality.}
Let $\GMLg$ denote the extension of $\GML$ with the global modalities for each $k \in \mathbb{N}$ (cf. Definition \ref{def:global-modality}). The proof of our characterization result for $\GMLg$ mirrors that of Theorem \ref{thm:gml-back-result}. Let $\sigma_G = \Prop \cup \Act \cup \{R_G\}$, where $R_G$ is a fresh action not in $\Act$. We write $\Diamond^{\geq k}_{G}$ for the graded modalities associated with the action $R_G$. The following definition is analogous to the ``backwards expansion'' (cf. Definition \ref{def:backward-expansion}).

\begin{definition}
Given a pointed $\sigma$-LTS $M_a$, the \emph{global expansion} of $M_a$ is the $\sigma_G$-expansion $M^G_a$ of $M_a$ given by setting $R^{M^G}_G = \dom(M) \times \dom(M)$.
\end{definition}

\noindent Recall that $\FKT^k_\sigma$ denotes the class of connected acyclic $\sigma$-LTSs of depth at most $k$, and that $\FKT_\sigma = \bigcup_{k \in \omega} \FKT^k_\sigma$ (cf. Definition \ref{def:classes}). The next definition defines a relation between structures in $\FKT_\sigma$ and those in $\KT_\sigma$. Its role is analogous to that of $\flip$ (cf. Definition \ref{def:flip}) in the proof of Theorem \ref{thm:gml-back-result}.

\begin{definition}
Let $T_c \in \FKT_\sigma$. We say that a $\sigma_G$-expansion $T'_c$ of $T_c$ is an \emph{$R_G$-connection of $T_c$} if $T'_c$ is a $\sigma_G$-tree and $T_c = T'_c \upharpoonright \sigma$.
\end{definition}

It's easy to see that every $T'_c \in \KT_{\sigma_G}$ is an $R_G$-connection of some $T_c \in \FKT_\sigma$. Similarly, for all $T_c \in \FKT_\sigma$, we have that $T_c = T'_c \upharpoonright \sigma$ (the reduct of $T'_c$ to the signature $\sigma$) for some $T'_c \in \KT_{\sigma_G}$. From these definitions, it is straightforward to prove the following analogues of Lemma \ref{lemma:gml-back-count-equiv} and Lemma \ref{lemma:gml-back-logic-equiv}.

\begin{lemma}
\label{lemma:glob-equiv}
Let $M_a$ and $N_b$ be finite $\sigma$-LTSs. Then
$$\homN(\FKT_\sigma,M_a) = \homN(\FKT_\sigma,N_b) \iff \homN(\KT_{\sigma_G},M^G_a) = \homN(\KT_{\sigma_G},N^G_b).$$
\end{lemma}

\begin{lemma}
\label{prop:glob-exp-lemma}
Let $M_a$ and $N_b$ be finite $\sigma$-LTSs. Then $M_a \GMLgequiv^\sigma N_b$ if and only if $M^G_a \GMLequiv^{\sigma_G} N^G_b$.
\end{lemma}

In the case of the global modality, we state our result only for finite $\sigma$-LTSs $M_a$. This is necessary, since homomorphisms out of $\sigma$-forests could map connected components which do not contain $c$ to any connected component in $M_a$, and so $\Hom(T_c,M_a)$ may be infinite even if $M_a$ is degree-finite.

\begin{theorem}
If $M_a$ and $N_b$ are finite $\sigma$-LTSs, then $M_a \GMLgequiv^\sigma N_b$ if and only if $\homN(\FKT_\sigma,M_a) = \homN(\FKT_\sigma,N_b)$.
\end{theorem}
\begin{proof}
By Lemma \ref{lemma:glob-equiv}, Theorem \ref{thm:gml-main}, and Lemma \ref{prop:glob-exp-lemma}.
\end{proof}


\section{Hybrid Logic}
\label{sec:hybrid-logic}
The hybrid logic $\HL$ is the extension of the basic modal language with the $\downarrow$-binder, the $@$-operator, and a countably infinite collection $\wvar$ of \emph{world variables} \cite{areces2001hybrid}. Formulas of $\HL$ are generated by the following grammar:
$$\varphi :=  p \mid x \mid \varphi \land \varphi \mid \varphi \lor \varphi \mid \lnot \varphi \mid \Diamond_i \varphi \mid \square_i \varphi \mid \downarrow x. \varphi \mid @_x \varphi,$$
where $i \in I$, $p \in \Prop$, and $x \in \wvar$\footnote{Readers familiar with $\HL$ should note that we omit \emph{nominals} from our presentation.}. A world variable $x$ occurs \emph{free} in a formula $\varphi$ if it does not occur in a subformula of $\varphi$ of the form $\downarrow x. \psi$, and \emph{bound} otherwise. A formula is a \emph{sentence} if it contains no free (world) variables.

An \emph{assignment} for a $\sigma$-LTS $M_a$ is a map $g: \wvar \to \dom(M)$. Given an assignment $g$, a world variable $x_i$, and a state $m \in \dom(M)$, we let $g[x_i\mapsto m]$ denote the assignment which is the same as $g$, except that it maps $x_i$ to $m$. The semantics (omitting the propositional, Boolean, and modal clauses, which are defined as usual) for $\HL$ are given as follows
\begin{center}
\begin{tabular}{l l l}
$M_a,g \models x$ &if $g(x) = a$ for $x \in \wvar$, \\
$M_a,g \models \downarrow x_i. \varphi$ &if $M_a,g[x_i\mapsto a] \models \varphi$, ~\text{and} \\
$M_a,g \models @_x \varphi$ &if $M_b,g \models \varphi$, where $g(x) = b$,
\end{tabular}
\end{center}
For $\HL$ sentences $\varphi$, the assignment chosen does not matter, and so we write $M_a \models \varphi$ instead of $M_a,g \models \varphi$. Given a $\sigma$-LTS $M_a$, the \emph{submodel of $M$ generated by $a$} is the structure $\gsub(M_a)$, defined to be the smallest substructure $M'_a$ of $M_a$ containing $a$ and such that, whenever $b \in \dom(M')$ and $R^M(b,c)$ holds, then $c \in \dom(M')$. Clearly $\gsub(M_a)$ is a point-generated $\sigma$-LTS. The following known result relates $\HL$ to the generated submodel-invariant fragment of $\FOL$, where a formula $\varphi$ is \emph{invariant for generated submodels} if, for any $\sigma$-LTS $M_a$, we have $M_a \models \varphi$ if and only if $\gsub(M_a) \models \varphi$.

\begin{theorem}[Generated Submodel Invariance, \cite{areces2001hybrid}]
\label{thm:gsub-inv}
If $\varphi(x)$ is a first-order formula in a modal signature, then $\varphi(x)$ is equivalent to a (nominal-free) $\HL$ sentence if and only if $\varphi(x)$ is invariant for generated submodels.
\end{theorem}

We write $\gsub^k(M_a)$ to denote the substructure of $\gsub(M_a)$ containing only elements of depth at most $k$. If $M_a$ is image-finite, then $\gsub^k(M_a)$ is finite for all $k \in \mathbb{N}$. The next proposition follows easily from Theorem \ref{thm:gsub-inv}.

\begin{proposition}
\label{prop:finite-gen-iso}
For each $N_b \in \PG^k_\sigma$, there is a formula $\varphi \in \HL$ such that, if $M_a$ is an image-finite point-generated $\sigma$-LTS, then $M_a \models \varphi$ if and only if $\gsub^k(M_a) \cong N_b$.
\end{proposition}
\begin{proof}
Fix some $N_b$ in $\PG^k$ with $\dom(N) = \{ b_1 \hdots b_n \}$, where $b = b_1$. Let $\delta(x_1,\hdots,x_n)$ be the $\FOL$ formula expressing that the $x_i$ are distinct, and that for all $y$, $y$ is reachable from $x_1$ by a directed $\sigma$-path of length at most $n$ if and only if $y = x_j$ for some $1 \leq j \leq n$. Consider the $\FOL$ formula
$$\psi(x_1) := \exists x_2 \hdots \exists x_n \left( \delta(x_1,\hdots,x_n) \land \left( \bigwedge_{i,j \leq n : R^N(b_i,b_j)} R(x_i,x_j) \right) \right).$$
If $M_a$ is an image-finite point-generated LTS, then $M \models \psi(a)$ if and only if $\gsub^k(M_a) \cong N_b$. Clearly $\psi(x_1)$ is a first-order formula in a modal signature which is invariant for generated submodels, and so there exists a nominal-free $\HL$ sentence equivalent to $\psi(x_1)$, which is what we wanted to show.
\end{proof}

\noindent The next lemma is obvious, since homomorphisms preserve path lengths and map distinguished elements to distinguished elements.

\begin{lemma}
\label{lemma:gsub-hom-count-preservation}
If $M_a$ is an image-finite LTS, then for each $k \in \mathbb{N}$, we have that $\homN(\PG^k_\sigma,M_a) = \homN(\PG^k_\sigma,\gsub^k(M_a))$.
\end{lemma}

\noindent We now prove our characterization result for $\HL$.

\begin{theorem}
\label{thm:hl-main}
For image-finite LTSs $M_a$ and $N_b$, the following are equivalent:
\begin{enumerate}
\item $\homN(\PG_\sigma,M_a) = \homN(\PG_\sigma,N_b)$,
\item $\gsub(M_a) \cong \gsub(N_b)$,
\item $M_a \HLequiv N_b$.
\end{enumerate}
\end{theorem}
\begin{proof}
For $(i)$ to $(ii)$, suppose that 
$\homN(\PG_\sigma,M_a) = \homN(\PG_\sigma,N_b)$. Then clearly $\homN(\PG^k_\sigma,M_a) = \homN(\PG^k_\sigma,N_b)$, and so by Lemma \ref{lemma:gsub-hom-count-preservation}, we have that $\homN(\PG^k_\sigma,\gsub^k(M_a)) = \homN(\PG^k_\sigma,\gsub^k(N_b))$. Hence by Fact \ref{prop:pg-ext-equivalence} and Theorem \ref{thm:AKW-lovasz}, $\gsub^k(M_a) \cong \gsub^k(N_b)$ for each $k \in \mathbb{Z}^+$, and so by Lemma \ref{lemma:iso-ext}, we have that $\gsub(M_a) \cong \gsub(N_b)$. For $(ii)$ to $(i)$, suppose that we have $\gsub(M_a) \cong \gsub(N_b)$. Since the range of a homomorphism from a point-generated $\sigma$-LTS to $M_a$ (resp. $N_b$) is contained within $\gsub(M_a)$ (resp. $\gsub(N_b)$), we have that $\homN(\PG_\sigma,M_a) = \homN(\PG_\sigma,N_b)$. The direction $(ii)$ to $(iii)$ is immediate from Theorem \ref{thm:gsub-inv}. For $(iii)$ to $(ii)$, Proposition \ref{prop:finite-gen-iso} gives us formulas $\varphi_k \in \HL$ such that $N_b \models \varphi$ if and only if $\gsub^k(N_b) \cong \gsub^k(M_a)$ for all $k \in \mathbb{N}$. Since $\gsub^k(M_a) \models \varphi$, we have by Theorem \ref{thm:gsub-inv} that $M_a \models \varphi$, and so by the assumption that $M_a \HLequiv^k N_b$, we have $N_b \models \varphi$. Hence $\gsub^k(M_a) \cong \gsub^k(N_b)$ for all $k \in \mathbb{N}$. Then by Lemma \ref{lemma:iso-ext}, $\gsub(M_a) \cong \gsub(N_b)$.
\end{proof}

\noindent We do not provide a version of this theorem which is parametrized by modal depth, as we did for $\GML$ (cf. Theorem \ref{thm:gml-main}), because Proposition \ref{prop:finite-gen-iso} does not offer a bound (as a function of $k$) on the modal depth of the $\HL$ formula describing a point-generated submodel of depth at most $k$ up to isomorphism.

\paragraph{Backward and Global Modalities.} $\HLe$, the extension of $\HL$ with the global modality, is known to have the expressive power of full first-order logic \cite{areces2007hybrid}, which we noted previously is captured by the left profile over the natural number semiring with respect to the class of all structures. This implies that $\HLe$ equivalence is captured by restricting the left profile over the natural semiring to the class of all $\sigma$-LTSs. We now provide a characterization result for $\HLb$, the extension of $\HL$ with the backward modalities for $k=1$ (cf. Definition \ref{def:backward-modalities}). As in Section \ref{sec:graded-ml}, we fix an expanded signature $\sigma_B = \Prop \cup \Act \cup \Act_B$, where $\Act_B = \{ B_i \mid i \in I \}$ is disjoint from $\Act$.

\begin{definition}
Let $T_c \in \C_\sigma$ and $T'_c \in \PG_{\sigma_B}$ with $\dom(T) = \dom(T')$. We say that $T'_c$ is a \emph{$\PG$-augmentation} of $M_a$ if, for each $i \in I$, there exists some $X_i \subseteq R_i^T$ such that $R_i^{T'} = R_i^T \setminus X_i$ and $B_i^{T'} = B_i^T \cup X_i^{-1}$.
\end{definition}

Thus $T'_c$ is a $\PG$-augmentation of $T_c$ if it can be obtained by replacing $R_i$ transitions in $T_c$ with $B_i$ transitions in the opposite direction. Recall the $\flip$ operation (cf. Definition \ref{def:flip}). If $T'_c$ is a point-generated $\sigma_B$-LTS, then clearly $\flip(T'_c)$ is a connected $\sigma$-LTS.

\begin{definition}
\label{def:expand}
Let $M_a$ be a connected $\sigma_B$-LTS. We write $\reach{M_a}$ to denote the set of elements in $\dom(M)$ reachable by a directed $\sigma_B$-path from $a$, and we set $\unreach{M_a} = \dom(M) \setminus \reach{M_a}$. If all transitions from elements of $\unreach{M_a}$ to elements of $\reach{M_a}$ are actions in $\Act$, then we write $\rec(M_a)$. Furthermore, if $\rec(M_a)$ is satisfied, then we define $\expand{M_a}$ to be the $\sigma_B$-LTS with $\dom(\expand{M_a}) = \dom(M)$ and $p^{\expand{M_a}} = p^M$ for all $p \in \Prop$, such that for all $i \in I$,
\begin{align*}
R_i^{\expand{M_a}} &= R_i^M \setminus \{ \langle m,n \rangle \mid \langle m, n \rangle \in R_i^M, m \in \unreach{M_a}, n \in \reach{M_a} \}, ~\text{and} \\
B_i^{\expand{M_a}} &= B_i^M \cup \{ \langle n,m \rangle \mid \langle m, n \rangle \in R_i^M, m \in \unreach{M_a}, n \in \reach{M_a} \}.
\end{align*}
\end{definition}

Intuitively, $\rec(M_a)$ asserts that all transitions out of elements of $\unreach{M_a}$ to elements of $\reach{M_a}$ are actions in $\Act$. The $\expfunc$ operation replaces $R_i$ transitions from elements of $\unreach{M_a}$ to elements of $\reach{M_a}$ with the corresponding $B_i$ actions in the opposite direction. The next proposition shows that $\expfunc$ is an operation on the class of connected $\sigma_B$-LTSs $M_a$ satisfying $\rec(M_a)$ which grows the set of elements reachable by a $\sigma_B$-path from $a$. The proof is straightforward.

\begin{proposition}
\label{prop:expand-facts}
For all $M_a \in \C_{\sigma_B}$ satisfying $\rec(M_a)$, we have that $\reach{M_a} \subseteq \reach{\expand{M_a}}$ and $\unreach{\expand{M_a}}_a \subseteq \unreach{M_a}$, and these inclusions are proper if $\unreach{M_a} \neq \emptyset$. Furthermore, $\expand{M_a}$ is a connected $\sigma_B$-LTS satisfying $\rec(M_a)$.
\end{proposition}

\begin{proposition}
\label{prop:pg-aux-exists}
If $T_c$ is in $\C_\sigma$, then there is a $\PG$-augmentation $T'_c$ of $T_c$.
\end{proposition}
\begin{proof}
Suppose $T_c$ is a finite connected $\sigma$-LTS. Since $T_c$ contains no $\sigma_B$ transitions, it clearly satisfies $\rec(T_c)$. Recall from Proposition \ref{prop:expand-facts} that $\unreach{\expand{T_c}}$ is a proper subset of $\unreach{T_c}$ whenever $\unreach{T_c} \neq \emptyset$, and so there exists some $k \in \mathbb{N}$ such that $\unreach{\expandk{T_c}} = \emptyset$, where the exponent $k$ indicates iterated application of the $\expfunc$ operation (which we can do by Proposition \ref{prop:expand-facts}). Hence $\reach{\expandk{T_c}} = \dom(\expandk{T_c})$, and so $\expandk{T_c}$ is a point-generated $\sigma_B$-LTS. Furthermore, $T'_c := \expandk{T_c}$ is clearly a $\PG$-augmentation of $T_c$.
\end{proof}

\noindent The following proposition is straightforward to prove from the definitions.

\begin{proposition}
\label{lemma:pg-pres}
Let $M_a$ be a degree-finite $\sigma$-LTS. Then
\begin{enumerate}
\item If $T_c$ is in $\C_\sigma$, then $\homN(T_c,M_a) = \homN(T'_c,M^B_a)$ for any $\PG$-augmentation $T'_c$ of $T_c$.
\item If $T'_c$ is in $\PG_{\sigma_B}$, then $\homN(\flip(T'_c),M_a) = \homN(T'_c,M^B_a)$.
\end{enumerate}
\end{proposition}

\begin{lemma}
\label{lemma:hybrid-back-equiv}
Let $M_a$ and $N_b$ be degree-finite $\sigma$-LTSs. Then $\homN(\C_\sigma,M_a) = \homN(\C_\sigma,N_b)$ if and only if $\homN(\PG_{\sigma_B},M^B_a) = \homN(\PG_{\sigma_B},N^B_b)$.
\end{lemma}
\begin{proof}
Both directions are by contraposition. For the reverse direction, if $\homN(\C_\sigma,M_a) \neq \homN(\C_\sigma,N_b)$, then there is a finite connected $\sigma$-LTS $T_c$ such that $\homN(T_c,M_a) \neq \homN(T_c,N_b)$. Then by Proposition \ref{prop:pg-aux-exists}, there exists a $\PG$-augmentation $T'_c$ of $T_c$. Then by Lemma \ref{lemma:pg-pres}, we have that $\homN(T'_c,M_a) \neq \homN(T'_c,N_b)$, and hence $\homN(\PG_{\sigma_B},M^B_a) \neq \homN(\PG_{\sigma_B},N^B_b)$. The forward direction is similar, using the $\flip$ function and Lemma \ref{lemma:pg-pres}.
\end{proof}

\noindent The proof of the next lemma is analogous to that of Lemma \ref{lemma:gml-back-logic-equiv}.

\begin{lemma}
\label{lemma:hybrid-back-to-reg}
Let $M_a$ and $N_b$ be degree-finite $\sigma$-LTSs. Then $M_a \HLbequiv^\sigma N_b$ if and only if $M^B_a \HLequiv^{\sigma_B} N^B_b$.
\end{lemma}

\begin{theorem}
\label{thm:hlb-main}
If $M_a$ and $N_b$ are degree-finite $\sigma$-LTSs, then $M_a \HLbequiv^\sigma N_b$ if and only if $\hom(\C_\sigma,M_a) = \hom(\C_\sigma,N_b)$.
\end{theorem}
\begin{proof}
By Lemma \ref{lemma:hybrid-back-equiv}, Theorem \ref{thm:hl-main}, and Lemma \ref{lemma:hybrid-back-to-reg}.
\end{proof}

\section{Negative Results}
\label{sec:neg-res}
Recall that $\BML$ denotes the basic (multi)modal language. \emph{Positive modal logic} (notation: $\PBML$) is the fragment of $\BML$ without negation. We now show that $\PBML$-equivalence and $\BML$-equivalence do not admit homomorphism count indistinguishability characterizations. For this, it will be convenient to work with the modal equivalence relations corresponding to these languages.

\begin{definition}
Let $M_a$ and $N_b$ denote $\sigma$-LTSs. A \textit{directed simulation} from $M_a$ to $N_b$ is a relation $Z \subseteq \dom(M) \times \dom(N)$ with $(a,b) \in Z$ such that
\begin{enumerate}[leftmargin=50pt]
\item[(prop$^-$)] If $(m,n) \in Z$, then $\marking^M_\sigma(m) \subseteq \marking^N_\sigma(n)$;
\item[(forth)] For each $i \in I$, if $(m,n) \in Z$ and there's $s \in M$ such that $R_i^M(m,s)$, then there's some $t \in N$ such that $R_i^N(n,t)$ and $(s,t) \in Z$; and
\item[(back)] For each $i \in I$, if $(m,n) \in Z$ and there's $t \in N$ such that $R_i^N(n,t)$, then there's some $s \in M$ such that $R_i^M(m,s)$ and $(s,t) \in Z$.
\end{enumerate}
If directed simulations from $M_a$ to $N_b$ and $N_b$ to $M_a$ exist, then we say that they are \emph{directed simulation equivalent} (notation: $M_a \dirsim N_b$). $Z$ is a \emph{bisimulation} between $M_a$ and $N_b$ if it also satisfies the stronger condition (prop) asserting that $\marking^M_\sigma(m) = \marking^N_\sigma(n)$ whenever $(m,n) \in Z$. If a bisimulation between $\sigma$-LTSs $M_a$ and $N_b$ exists, then they are \emph{bisimilar} (notation: $M_a \bisim N_b$).
\end{definition}

\noindent Directed simulation equivalence and bisimulation capture $\PBML$-equivalence and $\BML$-equivalence, respectively, over image-finite $\sigma$-LTSs.

\begin{theorem} (Directed Simulation Equivalence Invariance, \cite{kurtonina1997simulating})
\label{thm:dirsiminvariance}
For image-finite $\sigma$-LTSs $M_a$ and $N_b$, we have that $M_a \dirsim N_b$ if and only if $M_a \PBMLequiv N_b$.
\end{theorem}

\begin{theorem} (Bisimulation Invariance, \cite{van1976modal})
\label{thm:bisiminvariance}
For image-finite $\sigma$-LTSs $M_a$ and $N_b$, we have that $M_a \bisim N_b$ if and only if $M_a \MLequiv N_b$.
\end{theorem}

For equivalence relations $\sim$ and $\approx$ on $\sigma$-LTSs, if $M_a \sim N_b$ implies $M_a \approx N_b$, then we say that $\sim$ is \textit{finer} than $\approx$, and $\approx$ is \textit{coarser} than $\sim$. A function $f$ with $\dom(f) = \mathbb{N}$ is \emph{ultimately periodic} if there exist $P \in \mathbb{Z}^+$ and $L \in \mathbb{N}$ such that $f(n) = f(n+P)$ for all $n \geq L$. If $L$ and $P$ are the least integers such that the ultimate periodicity condition is satisfied, then we refer to the sequence $\langle f(0), \hdots, f(L-1) \rangle$ as the \emph{preperiod} of $f$, and we refer to the sequence $\langle f(L), \hdots, f(L+P-1) \rangle$ as the \emph{periodic segment} of $f$.

\begin{proposition}
\label{prop:periodic-counting}
Let $\mathcal{S} = \langle S, +_S, \cdot_S, 0_S, 1_s \rangle$ be a semiring such that $\cntS$ is not injective. Then $\cnt_\mathcal{S}$ is ultimately periodic. Furthermore, the preperiod and the periodic segment are disjoint, and there are no elements which occur more than once in either the preperiod or the periodic segment.
\end{proposition}
\begin{proof}
If $\rng(\cntS)$ is not injective, then there exists some least $m$ such that $\cntS(m) = \cntS(L)$ for some $L < m$. Let $P = m - L$. Observe that $\cntS(a + b) = \cntS(a) +_S \cntS(b)$ for all $a,b \in \mathbb{N}$, by associativity of addition in $\mathcal{S}$. We show by induction on $n \in \mathbb{N}$ that $\cntS(n) = \cntS(n+P)$ if $n \geq L$. If $n = L$, then $\cntS(n) = \cntS(m) = \cntS(n+P)$. Now suppose inductively that $\cntS(n) = \cntS(n+P)$. Then
\begin{align*}
\cntS(n+1) &= \cntS(n) ~+_S~ \cntS(1) \hspace*{20pt} 
&\text{(Assoc. of $+_S$)} \\
&= \cntS(n+P) ~+_S~ \cntS(1) &\text{(Inductive Hypothesis)} \\
&= \cntS(n+1+P). &\text{(Assoc. of $+_S$)}
\end{align*}
Hence $\cntS(n) = \cntS(n+P)$ for all $n \geq L$. By the above argument, the periodic segment begins with the first appearance of an element of $\mathcal{S}$ which occurs twice in $\rng(\cntS)$, and so the preperiod does not contain repeated elements. This also implies that the preperiod and period are disjoint. Finally, the fact that the periodic segment must also not contain any duplicate elements is clear, since successive elements are obtained by adding $1_S$, and so a duplicate element must mark the start of another repetition of the periodic segment.
\end{proof}

\begin{theorem}
\label{thm:neg-res}
Let $\mathcal{S} = \langle S, +_S, \cdot_S, 0_S, 1_S \rangle$ be an arbitrary semiring, and let $\sim$ denote any relation finer than directed simulation and coarser than bisimulation. There does not exist a class $\mathcal{C}$ of $\sigma$-LTSs such that, for all finite $\sigma$-LTSs $M_a$ and $N_b$, we have $\homS(\mathcal{C},M_a) = \homS(\mathcal{C},N_b)$ if and only if $M_a \sim N_b$.
\end{theorem}

\begin{proof}
Suppose toward a contradiction that some such class $\mathcal{C}$ exists. For $n \in \mathbb{Z}^+$, let $K^n_a$ denote the $\sigma$-LTS with $n$ states, distinguished element $a$, $p^{K^n} = \dom(K^n)$, and $R^{K^n} = \dom(K^n) \times \dom(K^n)$. Clearly $K^n_a \bisim K^{n'}_{a'}$ for all $n,n' \in \mathbb{Z}^+$. Furthermore, for all $\sigma$-LTSs $T_c$ with $\lvert \dom(T) \rvert = k$, every map $h: T_c \to K^n_a$ with $h(c) = a$ is a homomorphism, so $\lvert \Hom(T_c,K^n_a) \rvert = n^{k-1}$.

We first rule out that $\mathcal{C}$ contains only $\sigma$-LTSs $T_c$ with $\lvert \dom(T) \rvert = 1$. If it did, then for all $\sigma$-LTSs $S_d$, we have $\homS(T_c,S_d) = 1$ if and only if $\Hom(T_c,S_d) \neq \emptyset$. Consider the $\sigma$-LTSs in Figure \ref{fig:trivial-class-counterexample}. By homomorphic equivalence, $\Hom(T_c,M_a) \neq \emptyset$ if and only if $\Hom(T_c,N_b) \neq \emptyset$. Hence $\homS(\mathcal{C},M_a) = \homS(\mathcal{C},N_b)$. However, since $M_a \not\dirsim N_b$, this implies that $M_a \not\sim N_b$, contradicting our assumption about $\mathcal{C}$. Thus $\mathcal{C}$ must contain a structure $T_c$ with $\lvert \dom(T) \rvert = k + 1$ for some $k \in \mathbb{Z}^+$.

\begin{figure}[h]
\centering
\begin{tikzpicture}
\node (a) at (1,0) {$a$};
\node (x1) at (0,0) {$\bullet$};
\node (x2) at (2,0) {$\bullet$};
\node (b) at (4.3,0) {$b$};
\node (y1) at (5.3,0) {$\bullet$};
\node at (-0.5,0) {$\lnot p$};
\node at (2.4,0) {$p$};
\node at (5.7,0) {$p$};

\draw[->] (a) to (x1);
\draw[->] (a) to (x2);
\draw[->] (b) to (y1);

\node at (1,-0.9) {$M_a$};
\draw[decorate, decoration = {brace}] (2.75,-0.5) --  (-0.75,-0.5);

\node at (5,-0.9) {$N_b$};
\draw[decorate, decoration = {brace}] (6,-0.5) --  (4,-0.5);
\end{tikzpicture}
\caption{Homomorphically-equivalent $\sigma$-LTSs such that $M_a \not \bisim N_b$ and $M_a \not \dirsim N_b$.}
\label{fig:trivial-class-counterexample}
\end{figure}

We now claim that $\cntS$ is non-injective. If $\cntS$ were injective, then $\homS(T_c,K^1_a) = \cntS(1) \neq \cntS(2^k) = \homS(T_c,K^2_a).$
Since $K^1_a \bisim K^2_a$ (and hence $K^1_a \sim K^2_a$), this contradicts our assumption about $\mathcal{C}$. Thus we may assume that $\cntS$ is non-injective, and so by Proposition \ref{prop:periodic-counting}, it is ultimately periodic: there exist $P \in \mathbb{Z}^+$ and $L \in \mathbb{N}$ such that $\cnt_\mathcal{S}(n) = \cnt_\mathcal{S}(n+P)$ for all $n \geq L$. Let $\pi = \pi_0\hdots\pi_{P-1}$ denote the periodic segment of $\cntS$, and assume that $L$ and $P$ are minimal, so that, by Proposition \ref{prop:periodic-counting}, $\pi$ contains no duplicate elements. Figure \ref{fig:cntS-sequence} depicts the range of $\cntS$.

\begin{figure}[h]
\centering
\begin{tikzpicture}
    \node (s0) at (0,0) {$0_S$};
    \node (s1) at (0.75,0) {$1_S$};
    \node(sdots) at (1.4,0) {$\hdots$};
    \node(sN-1) at (3,0) {$\cntS(L-1)$};
    
    \node (pi0) at (4.75,0) {$\pi_0$};
    \node (pi1) at (5.4,0) {$\hdots$};
    \node (pi2) at (6.3,0) {$\pi_{P-1}$};
    \node (pi3) at (7.5,0) {$\pi_0$};
    \node (pi4) at (8.15,0) {$\hdots$};
    \node (pi5) at (9.05,0) {$\pi_{P-1}$};
    \node (pi6) at (10,0) {$\hdots$};
    
    

\end{tikzpicture}
\caption{The counting sequence in $\mathcal{S}$.}
\label{fig:cntS-sequence}
\end{figure}

\noindent We now distinguish several cases, deriving a contradiction in each.
\begin{enumerate}
\item If $0_S$ occurs in $\pi$, then $L = 0$ (i.e., $\cntS$ is \emph{purely} periodic), and $\pi_0 = 0_\mathcal{S}$. Then $\cntS(n) = \cntS(n \mod P)$ for all $n \in \mathbb{N}$. Hence we have that $\homS(T_c,K^P_a) = \cntS(P^k) = \cntS(P^k \mod P) = 0_\mathcal{S}$,
while $\hom(T_c,K^1_a) = \cntS(1) = 1_\mathcal{S}$. This implies that  $\homS(T_c,K^1_a) \neq \homS(T_c,K^P_a)$, which is a contradiction since $K^1_a \bisim K^P_a$.
\item If $1_S$ occurs in $\pi$ but $0_S$ does not, then $\pi_0 = 1_S$. Distinguish cases.
\begin{enumerate}
\item If $P = 1$, then for all $\sigma$-LTSs $S_d$, we have $\homS(T_c,S_d) = 1$ if and only if $\Hom(T_c,S_d) \neq \emptyset$. Consider the example in Figure \ref{fig:trivial-class-counterexample}: by homomorphic equivalence, we have $\homS(\mathcal{C},M_a) = \homS(\mathcal{C},N_b)$. This is again a contradiction, since $M_a \not \dirsim N_b$.
\item If $P > 1$, then $\cntS(0) = 0_S$, and $\cntS(n) = \pi_{((n-1) \mod P)}$ for $n > 0$. Then since $P^k-1 \mod P = P-1$, we have that $\homS(T_c,K^P_a) = \pi_{((P^K -1) \mod P)} = \pi_{P-1}$. Furthermore, since $\pi_0 = 1$, $P-1 \neq 0$, and the periodic segment contains no repeated elements, $\pi_{P-1} \neq 1_\mathcal{S}$. Hence $\homS(T_c,K^1_a) \neq \homS(T_c,K^P_a)$.
\end{enumerate}
\item If $1_\mathcal{S}$ does not occur in $\pi$, then $\homS(T_c,K^n_a) = n^k \neq 1_\mathcal{S}$ for $n$ sufficiently large. Hence we have $\homS(T_c,K^1_a) \neq \homS(T_c,K^n_a)$.
\end{enumerate}
Since we reach a contradiction in each case, no such class $\mathcal{C}$ can exist.
\end{proof}

\section{Discussion}
\label{sec:discussion}
Our positive characterization results, summarized in Figure \ref{fig:summary}, could also be seen as characterizations of certain modal equivalence relations, just as our negative result (Theorem \ref{thm:neg-res}) was. For example, image-finite LTSs are $\GML$-equivalent if and only if there exists a \emph{graded bisimulation} between them \cite{de2000note}. Similarly, two LTSs are equivalent with respect to nominal-free $\HL$ formulas if and only if there is an $\omega$-bisimulation between them \cite{areces2001hybrid}.

\begin{figure}[h]
\centering
\begin{minipage}{.49\linewidth}
\centering
\begin{tabular}{|c|c|}
\hline
\textbf{Language} & \textbf{Captured by} \\
\hline
\PML & $\hom_\mathbb{B}(\KT_\sigma,M_a)$ \\
\hline
\PMLb & $\hom_\mathbb{B}(\AKT_\sigma,M_a)$ \\
\hline
\PMLg & $\hom_\mathbb{B}(\FKT_\sigma,M_a)$ \\
\hline
\GML & $\hom_\mathbb{N}(\KT_\sigma,M_a)$ \\
\hline
\GMLb & $\hom_\mathbb{N}(\AKT_\sigma,M_a)$ \\
\hline
\end{tabular}
\end{minipage}%
\begin{minipage}{.49\linewidth}
\centering
\begin{tabular}{|c|c|}
\hline
\textbf{Language} & \textbf{Captured by} \\
\hline
\GMLg & $\hom_\mathbb{N}(\FKT_\sigma,M_a)$ \\
\hline
\HL & $\hom_\mathbb{N}(\GK_\sigma,M_a)$ \\
\hline
\HLb & $\hom_\mathbb{N}(\C_\sigma,M_a)$ \\
\hline
\PBML & None \\
\hline
\BML & None \\
\hline
\end{tabular}
\end{minipage}
\caption{Summary of Characterization Results.}
\label{fig:summary}
\end{figure}

\paragraph{Related work.}
An initial catalyst for investigating a left-profile characterization for $\GML$ was recent work by Barcelo et. al. showing that nodes of undirected graphs are indistinguishable by a special case of GNNs (aggregate-combine GNNs) if and only if they are graded modal equivalent \cite{barcelo2020logical}. Given that $\GML$ is a syntactic fragment of $\Ctwo$, and that both $\Ctwo$-equivalence and indistinguishability by GNNs can be captured by the restriction of the left homomorphism vector to the class of undirected trees, this result naturally suggested that a similar restriction to appropriate classes of trees should capture graded modal logic, as we have shown (cf. Theorem \ref{thm:gml-main}).

Sections \ref{sec:pos-ext-ml}, \ref{sec:graded-ml}, and \ref{sec:hybrid-logic} provide homomorphism count indistinguishability characterizations using model-theoretic methods. An important line of related work studies categorical generalizations of Lov\'{a}sz's original result; early work in this direction includes \cite{isbell1991some} and \cite{pultr1973isomorphism}. More recent work on \emph{game comonads} formalizes model-comparison games (such as the bisimulation game) in category-theoretic terms \cite{abramsky2017pebbling,abramsky2021relating}. These game comands can be used to derive homomorphism count indistinguishability results from general categorical results. For example, Theorem \ref{thm:sim-eq} is a consequence of a general categorical result proven in \cite{abramsky2021relating}. Similarly, a weaker version of Theorem \ref{thm:gml-main}, applying to finite structures, was obtained in \cite{dawar2021lovasz} using these methods. Categorical and topological arguments were used in \cite{reggio2022polyadic} to provide the first Lov\'{a}sz-style results for classes of infinite structures.

Early negative results pertaining to characterizations of logical equivalences via homomorphism count indistinguishability begin with \cite{atserias2021expressive}, in this case limited to negative results with respect to counting done in the Boolean and natural number semirings. In \cite{lichter2024limitations}, the authors show that equivalence with respect to \emph{linear-algebraic logic} cannot be captured by homomorphism count indistinguishability with respect to any class of graphs, both when counting is done in the natural numbers, and when counting is done in an arbitrary finite prime field. The present paper goes a step further, using the more general algebraic structure of semirings as the basis of counting for its negative results.

\paragraph*{Future work.}
Our combinatorial model-theoretic arguments for Theorems \ref{thm:sim-eq}, \ref{thm:gml-main}, and \ref{thm:hl-main} are analogous to earlier results for Lov\'{a}sz-style theorems. However, the method of lifting these results to extensions of the languages with backward or global modalities is, to the author's knowledge, a novel approach. One future avenue of research would be to generalize these methods to a categorical setting. Furthermore, while the aforementioned categorical work has provided interesting sufficient conditions for Lov\'{a}sz-style theorems, there is not yet a concise \emph{necessary} condition for a logic to admit such a result. Another interesting avenue of research would be to use the insight gained from our broad negative result in Theorem \ref{thm:neg-res} to identify such a condition. A last direction for future work is to identify modal relations captured by homomorphism indistinguishability with respect to \emph{finite} classes of LTSs; these are naturally related to the notion of \emph{homomorphism query algorithms} \cite{cate2024when}.

\paragraph*{Acknowledgements.} 
I sincerely thank Balder ten Cate for his invaluable advising during the development of much of this work, and I thank Scott Weinstein and Val Tannen for helpful comments on an earlier draft.


\bibliographystyle{aiml}
\bibliography{bib}

\end{document}